\documentclass[letter,11pt]{article}

\usepackage[english]{babel}
\usepackage[utf8x]{inputenc}
\usepackage[T1]{fontenc}
\usepackage{setspace}
\usepackage[sort&compress,square,numbers]{natbib}
\usepackage{tikz}
\usetikzlibrary{calc}

\usepackage[a4paper,top=3cm,bottom=2cm,left=3cm,right=3cm,marginparwidth=1.75cm]{geometry}

\usepackage{amsmath}
\usepackage{amsfonts}
\usepackage{amssymb}
\usepackage{amsthm}
\usepackage{graphicx}
\usepackage{authblk}
\usepackage[colorinlistoftodos]{todonotes}

\newtheorem{defn}{Definition}
\newtheorem{prop}{Proposition}

\title{Evaluating the Alignment of a Data Analysis between Analyst and Audience}
\author[1]{Lucy D'Agostino McGowan\footnote{Corresponding author email: mcgowald@wfu.edu}}
\author[2]{Roger D. Peng}
\author[3]{Stephanie C. Hicks}
\affil[1]{Department of Statistical Sciences, Wake Forest University}
\affil[2]{Department of Statistics and Data Sciences, University of Texas at Austin}
\affil[3]{Departments of Biostatistics and Biomedical Engineering, Johns Hopkins University}

\begin{document}
\maketitle



\begin{abstract}
A challenge that data analysts face is building a data analysis that is useful for a given consumer. Previously, we defined a set of principles for describing data analyses that can be used to create a data analysis and to characterize the variation between analyses. Here, we introduce a concept that we call the \textit{alignment} of a data analysis between the data analyst and a consumer. We define a successfully aligned data analysis as the matching of principles between the analyst and the consumer for whom the analysis is developed. In this paper, we propose a statistical model for evaluating the alignment of a data analysis and describe some of its properties. We argue that this framework provides a language for characterizing alignment and can be used as a guide for practicing data scientists and students in data science courses for how to build better data analyses. 
\end{abstract}


\noindent \textbf{Author Contributions}: LDM, SCH, and RDP equally conceptualized, wrote, and approved the manuscript. 

\noindent \textbf{Disclosures}: None.

\noindent \textbf{Acknowledgements}: The authors do not have any funding to acknowledge.

\clearpage

\section{Introduction}


In the practice of data science, a data scientist builds a data analysis to extract knowledge and insights from examining data \citep{tuke:1962, tukeywilk1996}. The discussion of how to build a data analysis often proceeds without explicit reference to an audience or consumer for whom the data analysis is being developed. Indeed there is much for a data analyst to consider on their own with respect to statistical techniques, visualization methods, data processing approaches, and computational algorithms that do not involve the needs or requirements of anyone in particular. However, a critical goal for many data analyses is to be useful or persuasive to another person. This consumer could range from simply the person doing the analysis to a much broader external audience.

A challenge for the data analyst then is to build a data analysis that is useful for the intended consumer while simultaneously adhering to generally accepted statistical principles for high quality analyses. One extreme would be for the data analyst to ignore the consumer and build the data analysis as they see fit, hoping that the consumer simply accepts the analysis upon presentation. Another extreme would involve the data analyst simply giving the consumer whatever they asked for (within the confines of statistical theory, of course) without the need for additional discussion. Neither extreme seems ideal, in general. Analysts who ignore the specific interests of the consumer risk being ignored in turn, rendering their analysis ultimately useless. Analysts who meet the consumer's specifications exactly miss an opportunity to educate the consumer on techniques or approaches that might better serve their needs. In either case, the analysis conducted is not as useful as it perhaps could have been.


One approach to addressing the challenge of building a useful analysis is to consider the consumer as part of the design of the analysis itself. The role of the consumer here is to contribute the definition of usefulness for a given analysis and specify their requirements for meeting that definition. Initially, the analyst may not completely agree with the definition of usefulness nor have a desire or ability to meet the requirements. There may be a lack of \textit{alignment} between the analyst and the consumer over what constitutes utility in the context of the analysis. Subsequently, a negotiation may occur between the analyst and consumer in order to come to agreement over what would make the analysis most useful and how the analyst might shape the analysis to achieve that goal.

This paper seeks to define the concept of alignment between a data analyst and a consumer and to discuss its role in the design of data analyses and in ensuring the usefulness of analytic outputs. Usefulness can mean a variety of things here, but typically, a useful analysis will influence a decision-making process in a scientific, business, or policy context. Analyses designed with a broad range of consumers in mind may serve a rhetorical purpose, aiming to convince an audience of a specific point or to inform them about an important issue. From the analyst's perspective, alignment can be thought of as a heuristic for guiding the organization of the elements of a data analysis and its presentation.


We start by leveraging a set of principles of data analyses that we previously introduced that can be used to guide the creation of a data analysis and to characterize the variation between data analyses~\citep{dm2022design}. These principles of data analysis are prioritized qualities or characteristics that are relevant to the analysis, as a whole or individual components. For a given data analysis, a data analyst can assign allocations to these principles to increase or decrease the amount of resources dedicated to these characteristics in a given data analysis. These allocations which can be highly influenced by outside constraints, such as time or budget. In this way, different allocations of the principles by the analyst can lead to different data analyses, all addressing the same underlying question~\citep{Silberzahn2018}. 

Next, we use this set of principles for data analysis to propose a framework for modeling the alignment of a data analysis that relies critically on the consumer for which the analysis is developed. In particular, as every data analysis has a consumer that views the analysis with their own preconceived notions, characteristics, and biases, we consider the allocations of the principles by both the analyst and the consumers, who may have a different perspective of how these various principles should be allocated for a given data analysis. Neither set of principles allocated by the analyst or the consumer is necessarily correct or incorrect. However, we previously hypothesized that how successfully aligned a data analysis is may depend on how well-matched the analyst's allocations are to the consumer's allocations for a given analysis~\citep{hickspeng2019-elementsandprinciples}. 

In this paper, we make these ideas more concrete and introduce a metric of quality evaluation that we call the \textit{alignment} of a data analysis between the analyst and the consumer. We define an aligned data analysis as the matching of allocated principles between the analyst and the consumer on which the analysis is developed. In the following sections, we formalize those ideas by proposing a statistical model and general framework for evaluating the alignment of a data analysis (Sections~\ref{sec:framework}--\ref{sec:statmodel}) and discuss the implications of this framework (Section \ref{sec:properties}). Finally, we present an example with real-world data~(Section~\ref{sec:data}) and discuss how this framework can be used as a guide for practicing data scientists and students in data science courses for how to build better data analyses (Section \ref{sec:discussion}). 

\subsection{Example: Initially Misaligned Analysis}

Consider two principles that may play a role in a data analysis: The reproducibility of the analysis~\citep{peng:2011} and how well-matched the data are to the analysis~\citep{dm2022design}. Suppose an analyst is in a organization where there is a request to do a quick analysis for a presentation in an internal meeting one week from now. The analyst, as a general matter, may feel it worthwhile to dedicate time and energy to ensure that the analysis is reproducible, even if it makes the analysis take longer to prepare and execute. In addition, the analyst notices that while the exact dataset for the analysis is not yet available, a dataset containing a similar surrogate measure is available now. Given that the analysis is needed relatively quickly, the analyst figures that the surrogate measure is sufficient. Meanwhile the person requesting the analysis (the consumer) would greatly prefer if more effort were taken to obtain the exact dataset needed for the analysis, and that less time be spent on making it reproducible, given that the analysis will likely be a one-time product. For this analyst and consumer pair, we have a mismatch on these two principles in that the analyst would prefer to devote more time to making the analysis reproducible relative to matching the data while the consumer would prefer that more effort be spent on getting the better data relative to making the analysis reproducible. 

%

\section{Components of Variation Principle Allocation}
\label{sec:framework}

As described above and our in previous work, we consider data analyses to be constructed in a manner guided by a set of $K$ design principles \citep{dm2022design} characterizing the data analysis. Specifically, we defined principles of data analysis as \textit{data-matching}, \textit{exhaustive}, \textit{skeptical}, \textit{second-order}, \textit{clarity}, and \textit{reproducible}. In this paper, we assume that the $i^{th}$ data analyst allocates a set of resources $\left\{\alpha_i^{(1)}, \ldots, \alpha_i^{(K)}\right\}$ to each principle $k=1,\ldots,K$ where each $\alpha_i^{(k)}>0$ for all~$k$. The sum of the total allocation across the $K$ principles is $\alpha_i^{(0)} = \sum_{k=1}^K \alpha_i^{(k)}$. For example, $\alpha_i^{(0)}$ could be thought of as the total amount of money or effort available to allocate to each principle. We provide an illustration of this in Section~\ref{sec:data}.

Data analyses are built to be viewed by a consumer, which can be an individual person or a group of people and can include the data analyst them self. For now, we will consider the consumer to be an individual person, other than the data analyst, and consider the case when a consumer is more than one individual in Section \ref{sec:group-audiences}. As such, the consumer has their own preferences for how resources should be allocated for the analyst building the data analysis. For a given data analysis, the consumer allocations will be denoted by the set $\left\{\omega_j^{(1)},\dots,\omega_j^{(K)}\right\}$ such that $\omega_j^{(k)}>0$ and $\omega_j^{(0)} = \sum_{k=1}^K \omega_j^{(k)}$. 

\subsection{Fixed Variation in Principle Allocations}
\label{sec:fixed-var}

We allow for the possibility that there will be variation in the allocation of the principles from analysis to analysis, for both analyst and consumer. Some of that variation can be characterized as fixed, while other variation may be best considered as random. From the analyst's perspective, some of the determinants of how a given principle may be allocated are:
\begin{enumerate}
    \item \textit{Analysis-specific Resources}. Considerations about computing resources, time, budget, personnel, and other such resources and analysis characteristics can often require that an analyst allocate more or less to certain principles for analysis. For example, analyses that must be conducted in a short amount of time may be limited in their ability to explore multiple competing hypotheses and exhibit low skepticism.
    \item \textit{Question Significance and Problem Characteristics}. The significance of the question being addressed with the data may play a role in determining principle allocations. Questions of high significance, for example, may require a high degree of transparency or reproducibility. Questions of lower significance may be done in a ``quick-and-dirty'' fashion; should the question's significance change in the future the analysis may need to be re-done with a different set of principle allocations.
    \item \textit{Field-specific Conventions}. Analysts are often members of a field from which they may have received their training (e.g. statistics, economics, computer science, bioinformatics). Each field develops conventions regarding how analyses in their field should be conducted and we characterize this using a field-specific mean value for a given principle. Tukey~\cite{tuke:1962} emphasized that in data analysis, there is a heavy emphasis on ``judgment'', one particular form of which is based upon the experience of members of a given field. We consider each analyst and each consumer to be a member of a \textit{field} or profession. Let $f_i\in\{1,\dots,F\}$ be the index into a set of $F$ fields or professions for analyst $i$. An analyst who belongs to field $f_i$ will be trained in the conventions of that field, which places a field-specific mean value $\lambda^{(k)}_{f_i}$ for a given principle~$k$ compared to some reference principle~$r$. 
    \item \textit{Analytic Product}. Depending on the analytic product that will ultimately be presented to the consumer (e.g. PDF document, web-based dashboard, executable R Markdown document), the analyst may determine that certain principles should receive a greater or lesser allocation. 
\end{enumerate}

Similarly, the consumer for whom the analysis is being developed will determine their principle allocations based on a variety factors, including their \textit{perception} of resources available to the analyst, their judgment of the significance of the question, their own field-specific conventions (assuming the consumer and the analyst are not members of the same field), and their perception of what the analytic product should contain.

\subsection{Random Variation in Principle Allocations}

The above-enumerated list describes some of the fixed factors that may drive variation in how various data analytic principles are allocated. However, there may be variation that is more random in nature. In particular, we consider the randomness as arising from sampling from a population of analysts or potential consumers trained in specific fields. Different analysts, presented with the exact same question and data, will likely allocate principles differently and hence produce different analyses based on their own personal characteristics. Similarly, different consumers, considering the same analytic product, will allocate principles differently and evaluate the alignment of the analysis differently.

An individual analyst~$i$ will deviate from their expected relative field-specific allocation by an amount $\delta_i^{(k)}$ which we think of as being randomly distributed with mean $0$ and finite variance. Therefore, the field-specific principle contribution for analyst~$i$ is $\lambda^{(k)}_{f_i} + \delta^{(k)}_{i}$ for principle $k$ compared to a reference principle $r$ in any given data analysis. Similarly, consumer $j$ who belongs to field $f_j$ will have a relative field-specific principle contribution of $\lambda^{(k)}_{f_j} + \eta^{(k)}_{j}$, where $\eta^{(k)}_{j}$ is randomly distributed with mean $0$ and finite variance.

\subsection{Stages of Analytic Design}

We conceive of the analytic design process as broadly occurring in a sequence of stages~(Figure~\ref{fig:stages}). At the first or \textit{baseline} stage the analyst and consumer independently allocate various principles based on their field-specific conventions and personal views of the analysis. Following this stage is the \textit{analytic negotiation} stage, where the analyst consumer discuss the proposed analysis and negotiate over how various principles will be allocated. Finally, in the \textit{resolution} stage, both analyst and consumer adjust their principle allocations based on the analytic negotiation. We indicate the analyst's principle adjustments in the resolution stage, i.e. the change in principle allocation from the baseline stage, as $\phi_i^{(k)}$. Similarly, the consumer's principle adjustments are indicated as as $\theta_j^{(k)}$.

In some cases, the analyst may not have an opportunity to interact directly with the consumer to negotiate the principle allocations. In those cases, the analyst may obtain some indirect understanding of the consumer's expectations for an analysis by doing some background research. Such background research may lead to an adjustment of the analyst's principle allocations via $\phi_i^{(k)}$. However, the consumer's allocations will remain unchanged in the resolution stage (i.e. $\theta_j^{(k)}=0$) because there was never any direct interaction with the analyst.

\begin{figure}[tbh]
\begin{tikzpicture}[->,shorten >=1pt,auto,node distance=5cm,
                    thick,main node/.style={rectangle,draw,font=\sffamily\Large\bfseries}]

  \node[main node] (analyst) {Analyst};
  \node[main node, below of=analyst] (consumer) {Consumer};

  \draw[dashed, -] (consumer.east) -- ++(10cm,0) coordinate (dashend);
  \draw[dashed, -] (analyst.east) -- ++(10.3cm,0) coordinate (dashend_);
  \draw[dashed, <->] (dashend) -- ++ (0, 2cm) coordinate (theta);
  \draw[dashed, <->] (dashend_) -- ++ (0, -0.7cm) coordinate (phi);
  
  \node[main node, right of=analyst, node distance=5cm, yshift=-2.5cm] (negotiation) {Analytic Negotiation};

  \node[main node, right of=negotiation, node distance=7cm, yshift=1.5cm] (analyst_res) {Analyst};
  \node[main node, below of=analyst_res, node distance=1.75cm] (consumer_res) {Consumer};

  \path[every node/.style={font=\sffamily\small}]
    (analyst) edge node {} (negotiation)
    (consumer) edge node {} (negotiation);

  \path[every node/.style={font=\sffamily\small}]
    (negotiation) edge node {} (analyst_res)
    (negotiation) edge node {} (consumer_res);


   \node[font=\sffamily\bfseries, yshift=-3.1cm] at ($(analyst)!0.5!(consumer)$)   {Baseline};
   \node[font=\sffamily\bfseries, yshift=-3.1cm] at ($(negotiation)$)   {Negotiation};
   \node[font=\sffamily\bfseries, yshift=-3.7cm] at ($(analyst_res)!0.5!(consumer_res)$) {Resolution};
   \node[font=\sffamily\bfseries, xshift=0.5cm, yshift=0.4cm] at (phi) {$\phi_i^{(k)}$};
   \node[font=\sffamily\bfseries, xshift=0.5cm, yshift=-1cm] at (theta) {$\theta_j^{(k)}$};
\end{tikzpicture}
\caption{Stages of analytic design.}
\label{fig:stages}
\end{figure}
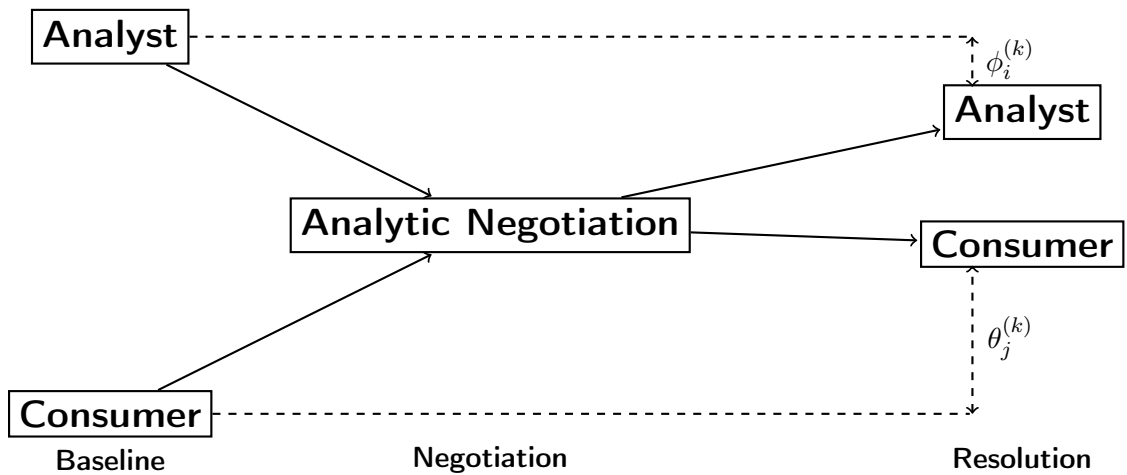

\section{Statistical Model for Principle Allocations}
\label{sec:statmodel}

Throughout text, we consider just a single data analysis at a time. For a given analysis and analyst $i$, the allocation assigned to a specific principle $k$ is $W_{i}^{(k)}$ and $\sum_{k=1}^K W_{i}^{(k)}=1$. Note here the allocations must sum to 1 across the $K$ principles, reflecting the reality that all analysts must decide how to allocate their priorities towards each principle when building a data analysis. We model the individual principle-specific allocations $W_{i}^{(k)}$ with the Dirichlet distribution:
\begin{equation}
\mathbf{W}_i 
=
\left(W_{i}^{(1)},\dots,W_{i}^{(K)}\right) 
\sim
\text{Dirichlet}\left(\alpha^{(1)}_{i},\dots,\alpha^{(K)}_{i}\right).
\label{eq:analyst}
\end{equation}
The parameters $\alpha_{i}^{(k)}$ are concentration parameters where $\alpha_{i}^{(k)}>0$. In the special case where the analyst gives the same allocation to all $K$ principles, the allocations would follow a symmetric Dirichlet distribution. For a given principle $k$, we can derive the marginal distribution from the Dirichlet and have
$$
W_{i}^{(k)}
\sim
\text{Beta}(\alpha_{i}^{(k)}, \alpha_i^{(0)}-\alpha_{i}^{(k)}).
$$
Where $\alpha_i^{(0)} = \sum_{k=1}^K\alpha_{i}^{(k)}$. We can then model the mean, $\mu_i^{(k)} = \alpha_i^{(k)}/\alpha_i^{(0)}$, as
\begin{equation}
\psi_{i}^{(k)}
=
\log\left(\frac{\mu^{(k)}_{i}}{\mu^{(r)}_{i}}\right)
=
\lambda_{f_i}^{(k)} + \delta_{i}^{(k)} + \phi_i^{(k)}\mathbf{1}\{\textrm{stage}=\textrm{resolution}\},
\label{eq:psi}
\end{equation}
where $r$ represents the reference principle, $\lambda_{f_i}^{(k)}$ is the field-specific coefficient for principle $k$ compared to principle $r$ for analyst $i$ in the field $f_i$, $\delta_{i}^{(k)}$ is analyst $i$'s deviation from the field for principle $k$, and $\phi_i^{(k)}$ represents the adjustment resulting from the analyst-consumer analytic negotiation. Here we would fit $K-1$ models where one of the $k$ principles is considered the reference ($r$) and therefore all coefficients in the model are equal to 0. We consider the analyst deviation $\delta_{i}^{(k)}$ to be randomly distributed across the population of potential analysts with mean $0$ and finite variance.

Analogous to the analyst's allocations, the allocation given to principle $k$ by consumer $j$ (who is a member of field $f_j$) can be written as $C_{j}^{(k)}$ with $\sum_{k=1}^K C^{(k)}_{j}=1$. We similarly model the vector $\mathbf{C}_j = 
\left(C_{j}^{(1)},\dots,C_{j}^{(K)}\right)$ as Dirichlet with concentration parameters $\omega_{j}^{(1)},\dots,\omega_{j}^{(K)}$.
We then similarly model the mean, $\mu_j^{(k)} = \omega_{j}^{(k)}/\omega_{j}^{(0)}$, as
\begin{equation}
\kappa^{(k)}_{j}
=
\log\left(\frac{\mu^{(k)}_{j}}{\mu^{(r)}_{j}}\right)
=
\lambda^{(k)}_{f_j} + \eta^{(k)}_{j} +\theta_j^{(k)}\mathbf{1}\{\textrm{stage}=\textrm{resolution}\}
\end{equation}
where $\lambda^{(k)}_{f_j}$ and $\eta_{j}^{(k)}$ are the field-specific coefficient and individual-level deviation from the field for the $j^{th}$ consumer, respectively, and $\theta_j^{(k)}$ represents the impact of the analyst-consumer analytic negotiation on the consumer's relative allocation allocation. Note that we consider $\eta_{j}^{(k)}$ to be independent of $\delta_{i}^{(k)}$ in the analyst's allocation model.

\subsection{Defining the Alignment of a Data Analysis}

With the analyst allocations in Equation~(\ref{eq:analyst}) and the consumer allocations, we can proceed to define the alignment of a data analysis.

\begin{defn}[Baseline Alignment]
The baseline alignment of a data analysis between analyst $i$ and consumer $j$ is defined in terms of the principle-specific allocation difference,
\begin{eqnarray}
B_{ij}^{(k)} & = & \left(\lambda_{f_i}^{(k)} - \lambda_{f_j}^{(k)}\right) + \left(\delta_{i}^{(k)}-\eta_{j}^{(k)}\right).
\label{eq:basedistance}
\end{eqnarray}
The overall baseline alignment for an analysis is defined via the collection of principle-specific allocation differences across all principles, $\mathbf{B}_{ij}=\left(B_{ij}^{(1)},\dots,B_{ij}^{(K)}\right)$.   
\label{def:baselinealign}
\end{defn}
The baseline alignment for an analysis represents the alignment that exists before the analyst and and consumer meet and negotiate any possible adjustments. Once we have arrived at the resolution stage, we can define the overall alignment for an analysis.

\begin{defn}[Overall Analysis Alignment]
For a given principle $k$, let $B_{ij}^{(k)}$ represent the baseline alignment between analyst $i$ and consumer $j$ for principle $k$. Let $\phi_i^{(k)}$ and $\theta_j^{(k)}$ represent allocation adjustments made by analyst $i$ and consumer $j$, respectively, upon reaching the resolution stage. Then the overall analysis alignment for principle $k$ is defined as
\begin{eqnarray}
D_{ij}^{(k)}
& = & 
B_{ij}^{(k)} + \left(\phi_i^{(k)} - \theta_j^{(k)}\right)\nonumber\\
& = &
B_{ij}^{(k)} + R_{ij}^{(k)}
\label{eq:distance}
\end{eqnarray}
The overall analyst-consumer alignment for a given data analysis is then characterized by the collection $D_{ij}^{(k)}$s for the entire set of $K$ principles $\mathbf{D}_{ij}=\left(D_{ij}^{(1)},\dots,D_{ij}^{(K)}\right)$. 
\label{def:overallalign}
\end{defn}
In Equation~(\ref{eq:distance}) above, $R_{ij}^{(k)}=\phi_i^{(k)}-\theta_j^{(k)}$ represents the ``residual alignment adjustment'' for principle $k$ made by the analyst and the consumer after considering their baseline difference.

In the next few definitions, we propose three ways to characterize the alignment of a data analysis pairwise between the analyst $i$ and consumer $j$: \textit{Strong Pairwise Alignment} (Definition~\ref{def:strongpairwise}) and \textit{Weak Pairwise Alignment} (Definition~\ref{def:weakpairwise}). 

\begin{defn}[Strong Pairwise Alignment]
A data analysis is strongly aligned between the pairing of analyst $i$ with consumer $j$ if for some small $\varepsilon > 0$,
$$
\left\|\mathbf{D}_{ij}\right\|_\infty 
= 
\max_{k=1,\dots,K} \left|D_{ij}^{(k)}\right| < \varepsilon.
$$
\label{def:strongpairwise}
\end{defn}
\noindent
The definition of strong pairwise alignment requires that the differences are never too large for any given principle.

We can propose a weaker form of analysis alignment that allows for some differences in how the principles are allocated, but places a limit on the total variation of those differences.
\begin{defn}[Weak Pairwise Alignment]
A data analysis is weakly aligned between the pairing of analyst $i$ with consumer $j$ if for some $p\geq 1$ and small $\varepsilon>0$,
\begin{equation}
\left\|\mathbf{D}_{ij}\right\|_p
=
\left(\frac{1}{K}\sum_{k=1}^K \left|D_{ij}^{(k)}\right|^p\right)^{1/p} < \varepsilon.
\label{eq:lpnorm}
\end{equation}
\label{def:weakpairwise}
\end{defn}
\noindent
With this definition, the analyst and consumer may differ slightly with respect to how each principle is allocated, but the overall differences between analyst and consumer must be small. The choice of $p$ here (and hence, the norm) will have an impact on how much deviation is allowed between analyst and consumer and how much any single principle may differ.

\subsection{Groups of Consumers}
\label{sec:group-audiences}
Up until this point we have assumed the consumer consisted of a single member indexed by $j$. However, it is common that a data analysis will be reviewed by or presented to a group of consumers. If there are $J$ consumers, then we can extend Equation~(\ref{eq:basedistance}) to be as follows.
\begin{eqnarray}
B_{i\cdot}^{(k)}
& = &
\frac{1}{J}\sum_{j=1}^J B_{ij}^{(k)}\nonumber \\
& = &
\psi^{(k)}_{i} - \frac{1}{J}\sum_{j=1}^J\kappa^{(k)}_{j}\nonumber \\
& = &
\left(\lambda_{f_i}^{(k)} - \frac{1}{J}\sum_j \lambda_{f_j}^{(k)}\right) + \left(\delta_{i}^{(k)}-\frac{1}{J}\sum_j\eta_{j}^{(k)}\right).
\label{def:basegroupalign}
\end{eqnarray}

Likewise, we can extend Equation~(\ref{eq:distance}) as
\begin{eqnarray}
D_{i\cdot}^{(k)}
& = &
\frac{1}{J}\sum_{j=1}^J B_{ij}^{(k)} + \left(\phi_i^{(k)} - \frac{1}{J}\sum_j\theta_j^{(k)}\right)\nonumber\\
& = &
B_{i\cdot}^{(k)} + R_{i\cdot}^{(k)}.
\label{def:groupalign}
\end{eqnarray}
In this formulation, $D_{i\cdot}^{(k)}$ is small if principle $k$ is allocated by the analyst in a manner that is equal to the mean of the group of consumers. With this extension of the principle-specific allocation difference, we can modify our definition of pairwise potential alignment to be defined for consumers with more than one member.

\section{Properties of Analytic Alignment}
\label{sec:properties}

In this section, we consider properties of alignment. The definitions of pairwise alignment and potential pairwise alignment presented in Section~\ref{sec:framework} lead to several implications about data analyses and what could potentially be done to improve the alignment of any given analysis.

\begin{prop}
        Let $\mathbf{R}_{ij}=\left(R_{ij}^{(1)},\dots,R_{ij}^{(K-1)}\right)$ be the vector of residual alignment adjustments for the $K$ principles. We say that the alignment of an analysis is improved from the baseline stage to the resolution stage (in the weak sense) if $\|\mathbf{D}_{ij}\|_2\leq\|\mathbf{B}_{ij}\|_2$. Let $\alpha$ be a scalar value and define $\mathbf{R}_{ij}=-\alpha\mathbf{B}_{ij}$. Then the alignment of the analysis at the resolution stage is improved if $0\leq\alpha\leq 2$.
\label{prop:alphaadjust}
\end{prop}
\begin{proof}
We can show directly that $\|\mathbf{D}_{ij}\|_2\leq\|\mathbf{B}_{ij}\|_2$ implies that
\begin{eqnarray*}
    \|\mathbf{B}_{ij} + \mathbf{R}_{ij}\|_2 
    & \leq & 
    \|\mathbf{B}_{ij}\|_2\\
    \|\mathbf{B}_{ij} - \alpha\mathbf{B}_{ij}\|_2
    & \leq & 
    \|\mathbf{B}_{ij}\|_2\\
    \|(1-\alpha)\mathbf{B}_{ij}\|_2
    & \leq & 
    \|\mathbf{B}_{ij}\|_2\\
    |1-\alpha|\|\mathbf{B}_{ij}\|_2
    & \leq &
    \|\mathbf{B}_{ij}\|_2
\end{eqnarray*}
The final statement can only be true if $0\leq\alpha\leq 2$.
\end{proof}

\begin{prop}
    The overall alignment of an analysis is optimized (in the weak sense) when $\|\mathbf{D}_{ij}\|_2=0$, which follows when $\mathbf{R}_{ij} = -\mathbf{B}_{ij}$.
\label{prop:alphaoptimal}
\end{prop}
\begin{proof}
    Follows directly from Proposition~\ref{prop:alphaadjust} by setting $\alpha=1$.
\end{proof}

Propositions~\ref{prop:alphaadjust} and~\ref{prop:alphaoptimal} suggest that the alignment of an analysis can be improved after the baseline stage by adjusting principle allocations in specific ways at the resolution stage. For example, if $B_{ij}^{(k)} > 0$, indicating that the analyst places greater allocation on principle $k$ (relative to the reference weight) than does the consumer, then the overall alignment can represent an improvement over the baseline alignment if the analyst sets $\phi_i^{(k)}=-0.5B_{ij}^{(k)}$ and the consumer sets $\theta_j^{(k)}=0.5B_{ij}^{(k)}$. This ``meet-in-the-middle" scenario represents one example of how to achieve $D_{ij}^{(k)}=0$ for a given principle, where the analyst reduces their allocation to principle $k$ and the consumer increases their allocation.

\begin{prop}[Field-matched Alignment]\label{prop:field} 
The overall alignment of an analysis is improved if either (1) the analyst and consumer are in the same field; or (2) the analyst and consumer have the same value for their respective field-specific means.
\end{prop}
\begin{proof}
If analyst $i$ and consumer $j$ have $f_i = f_j$, then we have $\lambda^{(k)}_{f_i} - \lambda^{(k)}_{f_j} = 0$ for all $k$. Similarly, if analyst and consumer are in different fields, alignment is improved if we have for the field-specific means $\lambda^{(k)}_{f_i} - \lambda^{(k)}_{f_j}=0$.
\end{proof}

The interpretation of Proposition~\ref{prop:field} is that members of the same field share similar conventions with respect to a given principle. For example, if ``computational reproducibility'' is the $k^{th}$ principle, then members of the field of computational biology (for example), which generally places a high allocation on computational reproducibility, might on average allocate more resources or effort to that principle. We might then expect data analyses in this field to generally demonstrate a high relative allocation on reproducibility, with perhaps code and data routinely made available. As a result, we would expect a higher potential for alignment (i.e. smaller $D_{ij}^{(k)}$) if analyst $i$ and consumer $j$ are both in the field of computational biology. Furthermore, it may be that even if analyst $i$ and consumer $j$ are not in the same field, their fields may share similar ideas about how much effort should be devoted to principle $k$.

\begin{prop}[Irreducible Baseline Principle Distance] 
$\|\mathbf{B}_{ij}\|>0$ with probability $1$ for all $i\ne j$.
\end{prop}
\begin{proof}
Because $\delta_i^{(k)}$ and $\eta_j^{(k)}$ are assumed to be continuous and random for all $k$, we have $\delta_i^{(k)} \ne 0$ and $\eta_j^{(k)}\ne 0$ with probability~$1$. Therefore, none of the $B_{ij}^{(k)}$ values can be equal to zero.
\end{proof}

Our model assumes that at the baseline stage, there will be some misalignment in principles due to at least individual differences between analyst and consumer. In defining strongly and weakly aligned analyses, we also allow for some differences between analyst and consumers at the resolution stage. The magnitude of allowable differences in principle allocations, $\varepsilon$, is likely to be analysis-specific and will depend in part on the context and circumstances surrounding the analysis. For a quick, ``work-in-progress'' type of analysis, the consumer may allow for larger deviations, with the presumption that the final version will have the appropriate principle allocations. More ``final'' analyses, such as a published paper, may require a stricter adherence to the consumer's principle allocations in order to optimize alignment.

\begin{prop}[Large audience effect]
Suppose the consumer is a large audience with $J$ members and that the analyst and the consumers are all in the same field (i.e. $f_i = f_j$ for $j=1,\dots,J$). Then the baseline alignment of an analysis for each principle is solely driven by the analyst's deviation from their field-specific mean.
\label{prop:largeaudience}
\end{prop}
\begin{proof}
If $J$ is very large, so that there are many audience members, then $\frac{1}{J}\sum_{j=1}^J \eta_j^{(k)}\approx 0$ because we assume that each consumer's individual deviation from their field-specific mean has mean~0. Furthermore, if the analyst and audience members are in the same field, then $\lambda_{f_i}=\lambda_{f_j}$ for all $j$. Therefore, the only remaining non-zero element in Equation~(\ref{def:basegroupalign}) is $\delta_i^{(k)}$, the analyst's individual-level deviation from the field-specific mean.
\end{proof}
When analyses will be designed for larger audiences, one way to achieve alignment is to minimize $\delta_i^{(k)}$ which effectively results in bringing the analyst's individual difference from the field-specific mean closer to zero. 

\section{Data Example}
\label{sec:data}

In this section we present some concrete examples of the implications of aspects of our model for analytic alignment on the design of data analyses. In particular, we discuss the impact of total resources dedicated to an analysis, potential categories of consumers and analysts, and present some real-world data describing how analysis designs can change from the baseline to the resolution stage.

\subsection{Impact of $\alpha_i^{(0)}$}

The total amount of resources available for analyst $i$ to allocate is represented by the parameter $\alpha_i^{(0)}$, and likewise the total amount of resources the consumer expects will be allocated to the analysis is represented by $\omega_j^{(0)}$. For the analyst, $\alpha_i^{(0)}$ is inversely proportional to the variance of the observed principle allocations such that larger $\alpha_i^{(0)}$ result in less random variation in principle allocations. In the context of a data analysis, this can be interpreted as providing more resources, e.g. time, money, etc., results in more consistent analyses from the prospective of design principles. 

To demonstrate the impact of $\alpha_i^{(0)}$, we have simulated two scenarios with three principles where the mean model representing the analyst's underlying principle allocations, specified as in $\psi_i^{(k)}$ in Equation~\ref{eq:psi}, remains the same, but the total amount of resources allocated, $\alpha_i^{(0)}$, varies. Figure~\ref{fig:alpha} displays this difference, with $\alpha_i^{(0)} = 1$ in Figure~\ref{fig:alpha}(a) and  $\alpha_i^{(0)} = 100$ in Figure~\ref{fig:alpha}(b). 

\begin{figure}
\includegraphics[width=5.8in]{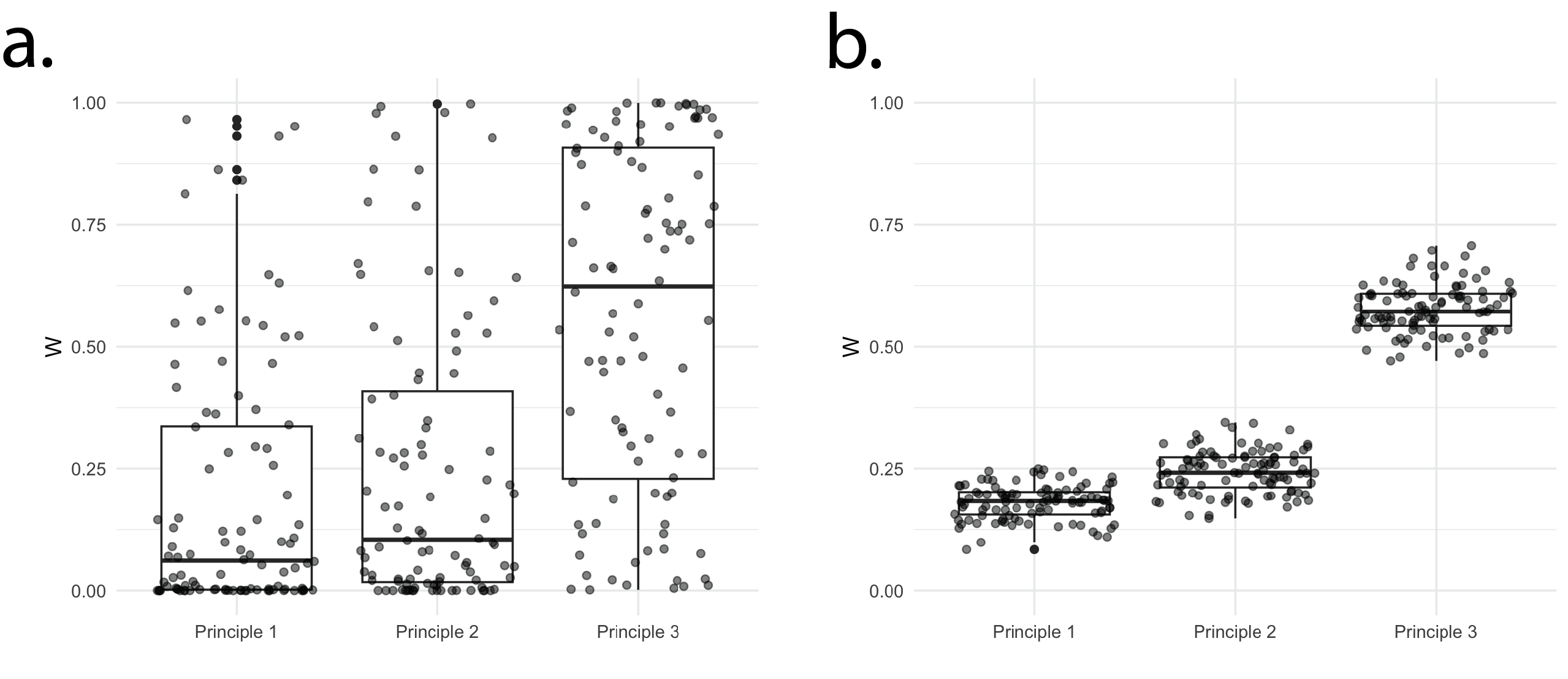}
\caption{Two simulated scenarios of observed principle allocations for an analyst with (a)~$\alpha_i^{(0)}=1$ and (b)~  $\alpha_i^{(0)}=100$. There is noticeable difference in the variability of observed allocations due to the differing $\alpha_i^{(0)}$ values.}
\label{fig:alpha}
\end{figure}

\subsection{Analyst - Consumer Scenarios}

We begin by presenting a series of scenarios to demonstrate the models discussed above. First, we will demonstrate three profiles of analyst/audience pairs: (1) accommodating analyst - intransigent consumer, (2) intransigent analyst - accommodating consumer, (3)  design-focused analyst - design-focused consumer. In all three scenarios, we reach alignment in expectation. In the first scenario, the driving force is the intransigent consumer, meaning they are inflexible with respect to updating their baseline expectations, in other words, $\theta_j^{(k)} = 0$ for $k = 1,\dots, K-1$. If the analyst is completely accommodating, meaning alignment is always achieved, $\phi_i^{(k)} = -\left[(\lambda_{fi}^{(k)} + \delta_i^{(k)})-(\lambda_{fj}^{(k)}+\eta_{j}^{(k)})\right]$ for $k = 1, \dots, K-1$. Figure~\ref{fig:bad-consumer} is an annotated example of this scenario for two principles, $a$ relative to $b$. In the second scenario, the analyst is intransigent and therefore $\phi_i^{(k)} = 0$ for $k = 1, \dots, K-1$ and $\theta_j^{(k)} = -\left[(\lambda_{fj}^{(k)} + \eta_j^{(k)})-(\lambda_{fi}^{(k)}+\delta_{i}^{(k)})\right]$ for $k = 1, \dots, K-1$. In the third scenario we describe both the analyst and consumer as \textit{design-focused}, by which we mean they are both motivated to improve alignment and are willing to negotiate over principle allocations. In this scenario, we can describe bounds for $\phi_i^{(k)}$ and $\theta_{j}^{(k)}$. In absolute value the analyst's post-resolution adjustment coefficient, $\phi_i^{(k)}$, is bounded by $\left|-\left[(\lambda_{fi}^{(k)} + \delta_i^{(k)})-(\lambda_{fj}^{(k)}+\eta_{j}^{(k)})\right]\right|$ and~0 with the qualification that the strict equality with 0 only occurs when alignment is achieved at baseline. Likewise, in absolute value the consumer's post-resolution adjustment coefficient, $\theta_j^{(k)}$, is bounded by $\left|-\left[(\lambda_{fj}^{(k)} + \eta_j^{(k)})-(\lambda_{fi}^{(k)}+\delta_{i}^{(k)})\right]]\right|$ and~0, again with the qualification that the strict equality with 0 only occurs when alignment is achieved at baseline.


\begin{figure}
\includegraphics[width=5.8in]{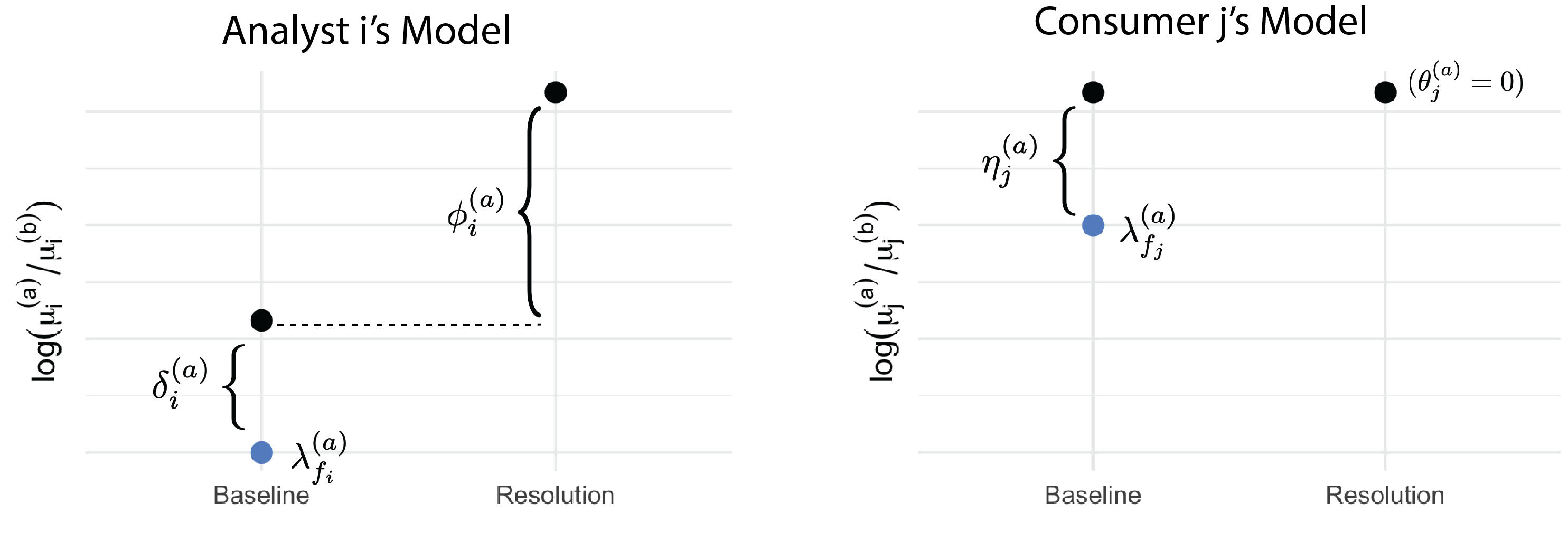}
\caption{Annotated display of analyst and consumer allocations pre- and post-negotiation for principle \textit{a} for the accommodating analyst - intransigent consumer pair. The left panel displays analyst $i$'s model, the right consumer $j$'s. In the left panel, the baseline value for the analyst is shown on the left, $\hat\lambda_{fi}^{({a})}+\hat\delta_i^{{(a)}}$. Then the resolution stage value after the analysts-consumer negotiation is displayed on the right, $\hat\lambda_{fi}^{({a})}+\hat\delta_i^{{(a)}} + \hat\phi_i^{({a})}$. The distance between the black point on the left to the point on the right represents $\hat\phi_i^{({a})}$. The distance between the top black point, analyst $i$'s baseline value for principle a, and the blue point, $\hat\lambda_{fi}^{({a})}$, is $\hat\delta_i^{({a})}$. The vertical distance between analyst $i$'s baseline value and the resolution stage value on the right represents $\hat\phi_i^{({a})}$. The bottom panel displays the model for consumer $j$, with the blue dot on the left representing $\lambda_{fj}^{(a)}$ and the black dot on the left representing  ($\lambda_{fj}^{(a)} + \eta_j^{(a)}$). There is no distance between the top black point on the left and the point on the right since this consumer is \textit{intransigent} and therefore $\theta_j^{(a)}=0$.}
\label{fig:bad-consumer}
\end{figure}

\subsection{Case Study}
\label{sec:casestudy}

In this section we describe a case study of analytic alignment. The data were collected from 26 students enrolled in a undergraduate capstone course at Wake Forest University. The study was approved by the Wake Forest University Institutional Review Board (IRB00025295). Participants were assigned to seven groups each tasked with completing a data analysis for a client. Students were taught six principles for designing a data analysis~\citep{dm2022design} (clarity, exhaustive, data-matching, reproducibility, second order, and skeptical) and given an initial problem statement from their assigned client. They were first asked to allocate each of the six design principles in the context of their assigned data analysis task to describe how much relative time they would allocate to each one, resulting in baseline analyst allocations of $\hat\lambda_{fi}^{(k)}+\hat\delta_{i}^{(k)}$ for all $i=1,...,26$ and $k=1,...,6$. For simplicity assume there were no additional covariates, $\mathbf{x}_i$ that contributed to the allocations. Students discussed these baseline allocations with their group, estimating the \textit{group-specific} allocation for each principle, $\hat\lambda_{f\cdot}$ for $f=1,\dots, 7$. These two steps fall in the \textit{baseline} stage of~\ref{fig:stages}. Groups were then asked to hold a meeting with their client where they discussed this proposed allocation and receive input on the client's expectations (the \textit{analytic negotiation}).
Finally, the students reported a post-negotiation set of allocations, taking into account both their initial allocation and the clients expectation, in the \textit{resolution} stage $\hat\lambda_{fi}^{(k)}+\hat\delta_{i}^{(k)}+\hat\phi_i^{(k)}$. 

Figure~\ref{fig:data} displays the results for all students across all groups and principle allocations. The data show variability both within and between groups, indicating differences in how each group and its members perceived the importance of the design principles before and after client negotiations. Assuming each group represents its member's "field", the differences within group at baseline are captured as part of the proposed mathematical process by $\delta_{i}^{(k)}$. The differences between groups may represent the differences in the analyses themselves, as each group was assigned a different analysis task. We then see that most groups updated their allocations post-negotiation, potentially improving pairwise alignment.

\begin{figure}
\includegraphics[width=5.8in]{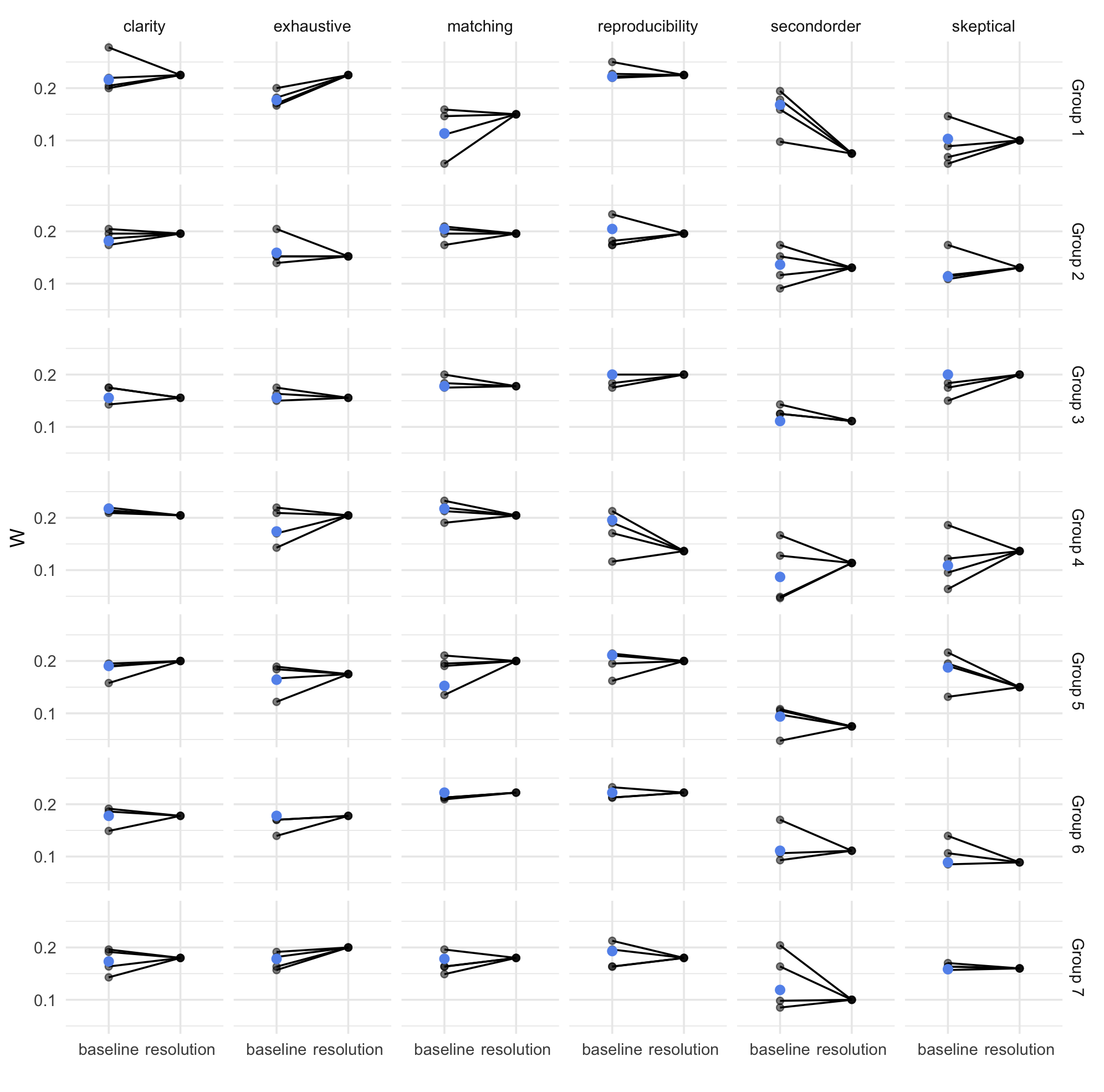}
\caption{Principle allocations before and after analytic negotiation. Each point represents the allocation assigned to the given principle by the analyst pre- and post-analytic negotiation. The blue point represents the agreed upon \textit{group-specific} allocation, a component of $\lambda_{f}^{(k)}$, which is the log relative allocation, as described in~\ref{sec:fixed-var}.}
\label{fig:data}
\end{figure}

\section{Discussion}
\label{sec:discussion}

In this paper we have developed a framework for characterizing the alignment of analysts and consumers in the development of a data analysis. This framework consists of a set of principles introduced earlier~\citep{dm2022design} and a statistical model that explains potential fixed and random variation in the allocation of these principles for a given analysis. We define the alignment of an analysis as the matching of allocations across principles between the analyst and consumer. Alignment can be obtained via negotiation between the analyst and consumer between the baseline and resolution stages of analytic planning.

A key point that we raise in this paper is the need for analyses to be designed properly in order to optimize their usefulness. Because the usefulness of an analysis varies between potential consumers of an analysis, there is a need to involve the consumer in the design of the analysis itself. Seasoned practitioners of data analysis will likely recognize this aspect of data analysis, as often there is an initial consultation that occurs where the specific requirements of an analysis are negotiated. We propose that such requirements can be driven by an underlying set of design principles and that negotiating agreement on how to allocate those principles to a given analysis is important for making an analysis useful. The data presented in Section~\ref{sec:casestudy} suggest that the alignment framework can serve as a teaching tool for students learning about the data analysis process.

The presentation we have taken is formal, mathematically speaking, because we believe the formalism can lead to valuable insights. For example, in Section~\ref{sec:framework} we described the notions of fixed and random variation between data analyses. Furthermore, it allows for a precise statement of what it means for an analysis to be aligned between consumer and analyst. Looking to the future, this mathematical formalism could provide a road map for developing approaches to estimating parameters we have specified from data, such as data analytic reports or papers. For example, the \textit{tidycode} package~\citep{tidycode2019} already allows for the automatic processing and classification of R code into data analysis activities using crowd-sourcing or pre-specified lexicons. Additionally, formalizing the mathematical framework for negotiation between analyst and audience allows us to leverage tools such as artificial intelligence, both for wide-spread pedagogical use as well as the study of how analysts and audiences interact in various environments. Our hope is that the mathematical model presented here serves as a solid foundation on which to build future knowledge about data analysis.

We argue that lack of alignment between analyst and consumer is an important failure mode for a data analysis and that consideration alignment in the early stages can greatly clarify the requirements of an analysis. Although there are numerous descriptions of failed data analyses~\citep{baggerly2009deriving}, specific definitions of how data analyses can fail, with detailed discussions of potential root causes, are lacking in the literature. Learning from failed data analyses is an important aspect of the training of any data analyst and the first step in that process is identifying when failure has occurred. Dialog between the analyst and the consumer about why a data analysis has failed can improve the quality of future analyses, as well as improve the quality of the relationship between analyst and consumer. Critical to such ``post-mortem'' discussions is that it be conducted in a blameless manner~\citep{park:2017} so that analyst and consumer can quickly come to a resolution over how problems should be fixed. While there is much previous work on assessments~\citep{Chance1997, Garfield2000, Chance2002} and project-based learning~\citep{Gnanadesikan1997, Smith1998, Tishkovskaya2012}, the literature provides little insight into how to recognize or quantitatively assess when, why and how a data analysis failure has occurred. Lack of alignment can serve as a specific factor to consider when addressing the failure of an analysis.

Our approach to characterizing data analysis failures shares many elements with the field of design thinking in its approach to building a solution matched to a specific consumer~\citep{cros:2011}. In some ways, one could think of a data analysis as a kind of ``product'', in the sense that it is not a naturally occurring object in nature. As such, someone---the analyst---must design the analysis in a manner that makes it useful to the consumer, or is aligned with the consumer's expectations and needs, much like any designed product. While the consumer could be one individual, or a group of individuals, each individual consumer plays a critical role in evaluating the quality of a given data analysis. Each consumer evaluates the quality with their own preconceived notions, characteristics, and biases towards valuing what makes a good or bad analysis~\citep{wildpfannkuch1999}. 

Our definition of alignment in data analysis depends solely on the participants---the analyst and the consumer---and the outputs of the data analysis. In theory, one could calculate the pairwise alignment of an analysis with just those elements. Critically, we do not consider events or information that occur outside the analysis or perhaps in the future. For example, an analysis may make certain conclusions based on the evidence available in the data that are later invalidated by more in-depth analysis (perhaps with better data). We do not therefore conclude that the original analysis was by definition a failure. At any given moment, an analysis can only draw on the data and evidence that are available. It therefore seems inappropriate to judge the alignment of a data analysis based on information that were not accessible at the time.


\clearpage 
\bibliographystyle{unsrtnat}
\bibliography{tods-alignment}

\begin{thebibliography}{17}
\providecommand{\natexlab}[1]{#1}
\providecommand{\url}[1]{\texttt{#1}}
\expandafter\ifx\csname urlstyle\endcsname\relax
  \providecommand{\doi}[1]{doi: #1}\else
  \providecommand{\doi}{doi: \begingroup \urlstyle{rm}\Url}\fi

\bibitem[Tukey(1962)]{tuke:1962}
J.~W. Tukey.
\newblock {The Future of Data Analysis}.
\newblock \emph{The Annals of Mathematical Statistics}, 33\penalty0
  (1):\penalty0 1--67, 1962.

\bibitem[Tukey and Wilk(1966)]{tukeywilk1996}
W.~Tukey and M.~B. Wilk.
\newblock Data analysis and statistics: An expository overview.
\newblock In \emph{{In Proceedings of the November 7-10, 1966, fall joint
  computer conference}}, pages 695--709, 1966.

\bibitem[{D'Agostino McGowan} et~al.(2022){D'Agostino McGowan}, Peng, and
  Hicks]{dm2022design}
Lucy {D'Agostino McGowan}, Roger~D Peng, and Stephanie~C Hicks.
\newblock Design principles for data analysis.
\newblock \emph{Journal of Computational and Graphical Statistics}, pages 1--8,
  2022.

\bibitem[Silberzahn et~al.(2018)Silberzahn, Uhlmann, Martin, Anselmi, Aust,
  Awtrey, Bahnik, Bai, Bannard, Bonnier, Carlsson, Cheung, Christensen, Clay,
  Craig, Rosa, Dam, Evans, Cervantes, Fong, Gamez-Djokic, Glenz, Gordon-McKeon,
  Heaton, Hederos, Heene, Mohr, Hogden, Hui, Johannesson, Kalodimos,
  Kaszubowski, Kennedy, Lei, Lindsay, Liverani, Madan, Molden, Molleman, Morey,
  Mulder, Nijstad, Pope, Pope, Prenoveau, Rink, Robusto, Roderique, Sandberg,
  Schluter, Schonbrodt, Sherman, Sommer, Sotak, Spain, Sporlein, Stafford,
  Stefanutti, Tauber, Ullrich, Vianello, Wagenmakers, Witkowiak, Yoon, and
  Nosek]{Silberzahn2018}
R.~Silberzahn, E.~L. Uhlmann, D.~P. Martin, P.~Anselmi, F.~Aust, E.~Awtrey,
  A.~Bahnik, F.~Bai, C.~Bannard, E.~Bonnier, R.~Carlsson, F.~Cheung,
  G.~Christensen, R.~Clay, M.~A. Craig, A.~Dalla Rosa, L.~Dam, M.~H. Evans,
  I.~Flores Cervantes, N.~Fong, M.~Gamez-Djokic, A.~Glenz, S.~Gordon-McKeon,
  T.~J. Heaton, K.~Hederos, M.~Heene, A.~J.~Hofelich Mohr, F.~Hogden, K.~Hui,
  M.~Johannesson, J.~Kalodimos, E.~Kaszubowski, D.~M. Kennedy, R.~Lei, T.~A.
  Lindsay, S.~Liverani, C.~R. Madan, D.~Molden, E.~Molleman, R.~D. Morey, L.~B.
  Mulder, B.~R. Nijstad, N.~G. Pope, B.~Pope, J.~M. Prenoveau, F.~Rink,
  E.~Robusto, H.~Roderique, A.~Sandberg, E.~Schluter, F.~D. Schonbrodt, M.~F.
  Sherman, S.~A. Sommer, K.~Sotak, S.~Spain, C.~Sporlein, T.~Stafford,
  L.~Stefanutti, S.~Tauber, J.~Ullrich, M.~Vianello, E.-J. Wagenmakers,
  M.~Witkowiak, S.~Yoon, and B.~A. Nosek.
\newblock {Many Analysts, One Data Set: Making Transparent How Variations in
  Analytic Choices Affect Results}.
\newblock \emph{Advances in Methods and Practices in Psychological Science},
  1\penalty0 (3):\penalty0 337--356, 2018.

\bibitem[Hicks and Peng(2019)]{hickspeng2019-elementsandprinciples}
S.~C. Hicks and R.~D. Peng.
\newblock {Elements and Principles of Data Analysis}.
\newblock \emph{arXiv}, pages 1--13, 2019.
\newblock URL \url{https://arxiv.org/abs/1903.07639}.

\bibitem[Peng(2011)]{peng:2011}
R.~D. Peng.
\newblock {{R}eproducible research in computational science}.
\newblock \emph{Science}, 334\penalty0 (6060):\penalty0 1226--1227, 12 2011.

\bibitem[{D'Agostino McGowan}(2019)]{tidycode2019}
L.~{D'Agostino McGowan}.
\newblock \emph{{tidycode: Analyze Lines of R Code the Tidy Way}}, 2019.
\newblock URL \url{https://CRAN.R-project.org/package=tidycode}.
\newblock R package version 0.1.1.

\bibitem[Baggerly and Coombes(2009)]{baggerly2009deriving}
Keith~A Baggerly and Kevin~R Coombes.
\newblock Deriving chemosensitivity from cell lines: Forensic bioinformatics
  and reproducible research in high-throughput biology.
\newblock \emph{The Annals of Applied Statistics}, pages 1309--1334, 2009.

\bibitem[Parker(2017)]{park:2017}
H.~Parker.
\newblock {Opinionated analysis development}.
\newblock \emph{PeerJ Preprints}, 5:\penalty0 e3210v1, 2017.
\newblock URL \url{https://peerj.com/preprints/3210/}.

\bibitem[Chance(1997)]{Chance1997}
B.~L. Chance.
\newblock {Experiences with Authentic Assessment Techniques in an Introductory
  Statistics Course}.
\newblock \emph{Journal of Statistics Education}, 5\penalty0 (3), 1997.

\bibitem[J. and Chance(2000)]{Garfield2000}
Garfield J. and B.~Chance.
\newblock {Assessment in Statistics Education: Issues and Challenges}.
\newblock \emph{Mathematical Thinking and Learning}, 2\penalty0 (1-2):\penalty0
  99--125, 2000.

\bibitem[Chance(2002)]{Chance2002}
B.~L. Chance.
\newblock {Components of Statistical Thinking and Implications for Instruction
  and Assessment}.
\newblock \emph{Journal of Statistics Education}, 10\penalty0 (3), 2002.

\bibitem[Gnanadesikan et~al.(1997)Gnanadesikan, Scheaffer, Watkins, and
  Witmer]{Gnanadesikan1997}
M.~Gnanadesikan, R.~L. Scheaffer, A.~E. Watkins, and J.~A. Witmer.
\newblock {An Activity-Based Statistics Course}.
\newblock \emph{Journal of Statistics Education}, 5\penalty0 (2), 1997.

\bibitem[Smith(1998)]{Smith1998}
G.~Smith.
\newblock {Learning Statistics by Doing Statistics}.
\newblock \emph{Journal of Statistics Education}, 6\penalty0 (3), 1998.

\bibitem[Tishkovskaya and Lancaster(2012)]{Tishkovskaya2012}
S.~Tishkovskaya and G.~A. Lancaster.
\newblock {Statistical Education in the 21st Century: A Review of Challenges,
  Teaching Innovations and Strategies for Reform}.
\newblock \emph{Journal of Statistics Education}, 20\penalty0 (2), 2012.

\bibitem[Cross(2011)]{cros:2011}
N.~Cross.
\newblock \emph{{Design thinking: {U}nderstanding how designers think and
  work}}.
\newblock Berg, 2011.

\bibitem[Wild and Pfannkuch(1999)]{wildpfannkuch1999}
C.~J. Wild and M.~Pfannkuch.
\newblock {Statistical Thinking in Empirical Enquiry}.
\newblock \emph{International Statistical Review}, 67\penalty0 (3):\penalty0
  223--248, 1999.

\end{thebibliography}

\end{document}